\documentclass{article}

\newcommand{\ignore}[1]{}
\newcommand{\remove}[1]{}

\usepackage{amsthm}

\usepackage[T1]{fontenc}
\usepackage[utf8]{inputenc}
\usepackage{amssymb,amsfonts,amsmath}
\usepackage{algorithm}
\usepackage[noend]{algpseudocode}
\usepackage{float}
\newfloat{algorithm}{t}{lop}

\usepackage{scrextend}
\addtokomafont{labelinglabel}{\sffamily\bfseries}
\usepackage{csquotes}
\newtheorem{criterion}{Criterion}
\newtheorem{theorem}{Theorem}[section]
\newtheorem{lemma}[theorem]{Lemma}

\usepackage{graphicx}

\usepackage{subfigure}
\usepackage{enumitem}

\usepackage{flushend}

\usepackage{fancyhdr}
\usepackage[normalem]{ulem}
\usepackage[hyphens]{url}
\usepackage{hyperref}

\usepackage{ifpdf}
\ifpdf
	\PassOptionsToPackage{pdftex}{graphicx}
	\usepackage{epstopdf}
\else
	\PassOptionsToPackage{dvips}{graphicx}
\fi

\usepackage{amsmath}
\usepackage{amssymb}

\begin{document}

%%%%%%%%%%%---SETME-----%%%%%%%%%%%%%
%\title{A Wait-free Multi-word Atomic (1,N) Register\\Exploiting RMW Machine Instructions}
\title{
%Anonymous Readers Counting:
A Wait-free Multi-word Atomic (1,N) Register for Large-scale Data Sharing on Multi-core Machines}
%%%%%%%%%%%%%%%%%%%%%%%%%%%%%%%%%%%%

\author{
Mauro Ianni (\texttt{mianni@dis.uniroma1.it}) \\
\multicolumn{1}{p{.8\textwidth}}{\centering\emph{DIAG - Sapienza University of Rome}} \\
\and
Alessandro Pellegrini (\texttt{pellegrini@dis.uniroma1.it}) \\
\multicolumn{1}{p{.8\textwidth}}{\centering\emph{DIAG - Sapienza University of Rome}} \\
\and
Francesco Quaglia (\texttt{francesco.quaglia@uniroma2.it}) \\
\multicolumn{1}{p{.8\textwidth}}{\centering\emph{DICII - University of Rome Tor Vergata}} \\
}

%\author{\IEEEauthorblockN{Mauro Ianni, Alessandro Pellegrini, Francesco Quaglia}
%\IEEEauthorblockA{DIAG--Sapienza Universit\`a di Roma, Italy}
%\\
%Email: emalele1688@gmail.com, mauroianni@gmail.com, \{pellegrini,quaglia\}@dis.uniroma1.it}
%}

%\begin{document}
\maketitle
%\thispagestyle{firstpage}
%\pagestyle{plain}

%%%%%% -- PAPER CONTENT STARTS-- %%%%%%%%
\begin{abstract}
We present a multi-word atomic (1,N) register for multi-core machines exploiting Read-Modify-Write (RMW) instructions to coordinate the writer and the readers in a wait-free manner. Our proposal, called Anonymous Readers Counting (ARC), enables large-scale data sharing by admitting up to $2^{32}-2$ concurrent readers on off-the-shelf 64-bits machines, as opposed to the most advanced RMW-based approach which is limited to 58 readers. Further, ARC avoids multiple copies of the register content when accessing it---this affects classical register's algorithms based on atomic read/write operations on single words. Thus it allows for higher scalability with respect to the register size. Moreover, ARC explicitly reduces
%the overall energy consumption,
improves performance via a proper limitation of RMW instructions in case of read operations, and by supporting constant time for read operations and amortized constant time for write operations.
%Our proposal has therefore a strong focus on real-world off-the-shelf architectures, allowing us to capture properties which benefit %both performance and energy consumption.
A proof of correctness of our register algorithm is also provided, together with experimental data for a comparison with literature proposals. Beyond assessing ARC on physical platforms, we carry out as well an experimentation on virtualized infrastructures, which shows the resilience of wait-free synchronization as provided by ARC with respect to CPU-steal times, proper of more modern paradigms such as cloud computing.
\end{abstract}

\section{Introduction}

Hardware-based atomicity facilities offered by multi-core computing platforms to manage single-word shared-objects are not sufficient to automatically guarantee atomicity when concurrent threads manipulate
multi-word objects. Synchronization algorithms are therefore needed to enable atomic read/write operations on this type of objects.
Also, the extreme level of scale-up of modern computing platforms, with projection towards exascale computing, demands for shared-object management algorithms that are capable of efficiently supporting huge levels of concurrency.

In this article we face such an issue by providing a pragmatic design and implementation of a shared-object algorithm in multi-processor/multi-core shared-memory machines. Specifically, we present Anonymous Readers Counting (ARC), which is an atomic (1,N)---one writer, N readers---register of arbitrary length (i.e., made up by an arbitrary number of words, which can change over time, possibly upon each update of the register) exhibiting the following features:
\begin{itemize}[noitemsep]
\item it is devised for a huge scale-up of the number of concurrent threads to be managed;
\item it targets the optimization of the actual execution path of the threads along multiple dimensions: locality, time complexity and actual cost of machine instructions to be executed.
\end{itemize}

We emphasize that providing optimized (1,N) registers is a relevant objective since they constitute building blocks to realize more general (M,N) registers, as already shown by several works (see, e.g.,~\cite{Lix96}).

As the core property enabling scalability, ARC guarantees {\em wait-freedom} of both write and read operations. In fact, it uses no locking scheme, and guarantees that no operation (either a read or a write) fails and no retry-cycles are ever needed. This is achieved by relying on Read-Modify-Write (RMW) atomic instructions offered by conventional Instruction Set Architectures (ISAs), which are exploited to manipulate meta-data that are used by concurrent threads to coordinate themselves when performing register operations.

A close literature proposal based on RMW instructions, which still guarantees wait-freedom of read/write operations on (1,N) registers, is the one in~\cite{Lar09}.
However, this proposal allows up to 58 readers only on conventional 64-bit machines, while ARC can manage up to $2^{32}-2$ readers, thus enabling a huge scale-up in the level of concurrency, as already hinted. Also, the approach in~\cite{Lar09} is based on deterministically forcing synchronization (via RMW instructions) upon any read operation, even in scenarios where the register's content has not been modified by the writer since the last read by the reader. ARC avoids executing RMW instructions in such situations, since it detects whether the last accessed snapshot of the register is still consistent (it is the most up to date one within the linearizable history of read/write accesses) by only relying on conventional memory-read machine instructions.

The effect of this optimization on the overall performance %and energy efficiency
is non-minimal, as we also show experimentally, given that RMW instructions pay anyhow a cost
%(both in terms of latency and energy consumption)
due to the effects on the interconnection among CPU-cores. For example, modern Intel-based architectures relying the QuickPath Interconnect~\cite{Zia10} require (and pay the cost of) message passing among CPU-cores when executing RMW instructions.
There are corner cases, like in the case of a split lock, where all cores might be required to enter a spin phase (via a message passing protocol) to ensure data coherency.
Furthermore, these effects can be amplified when a memory location updated by a RMW instruction is split across different cache lines, as shown in~\cite{Saf09}.

As opposed to more historical solutions for wait-free atomic (1,N) registers in shared-memory platforms~\cite{Lam77}, which only exploit atomic read/write operations (not RMW instructions) of individual memory words, we avoid multiple copies of the register content when performing either read or write operations. This allows for better scalability of ARC with respect to the size of the register content.

Still related to buffer management, ARC adheres to the classical lower bound of $N+2$ buffers~\cite{Lam86} keeping the different snapshots of the (1,N) register content, to be accessed in wait-free manner in some linearizable execution of read/write operations by the concurrent threads. Overall, compared to literature proposals, we enable definitely scaled up amounts of concurrent readers with no increased memory footprint and by not imposing extra memory-copy operations, which leads ARC to favor locality.

Furthermore, ARC allows constant time for read operations, jointly guaranteing amortized constant-time for write operations. This is not guaranteed by the RMW-based approach in~\cite{Lar09}, since it requires $O(N)$ time for write operations---an aspect that is clearly related to the reduced amount of readers admitted by this register algorithm.

Beyond presenting ARC, we also provide a proof of its correctness. Further, we report experimental data showing the benefits  from our proposal compared to a few existing literature solutions. Performance data have been collected on a physical 32 CPU-cores machine and on a virtual platform hosted by Amazon equipped with 40 vCPUs. This allows us to asses as well the resilience of our wait-free register algorithm on top of architectures which could suffer much more by lock-based or non-efficient non-blocking synchronization strategies, due to the fact that a delay imposed on a certain processor by physical unavailability of processing resources could in its turn affect the speed of all other processors, independently of the assigned computing power. As a last note, our experimental evaluation has been based on user-space code implementing ARC, but nothing prevents ARC to be integrated within lower-level software layers (such as an operating system kernel).

The remainder of this article is organized as follows. In Section \ref{sec:related} we discuss related work. ARC is presented in Section \ref{sec:technical}. Its correctness proof is provided in Section \ref{sec:proof}. Experimental results are reported in Section \ref{sec:experimental}.

\section{Related Work}
\label{sec:related}

%As hinted,
Our target are single-writer/multiple-readers shared-objects in multi-processor/multi-core machines, to be managed in a wait-free manner.  According to~\cite{Her91}, wait-freedom allows any concurrent operation on the shared-object to execute in a finite number of steps, regardless of any action carried out by other concurrent operations. This is not guaranteed neither by classical lock-based synchronization schemes~\cite{Sil94} nor by lock-free ones~\cite{Bar93,Her90}. Wait-freedom appears as a mandatory means to efficiently handle concurrent operations on shared-objects in systems with large/huge amounts of concurrent threads. In fact, it enables performance to be resilient to degradation with respect to the level of concurrency in the operations on the shared-object.

A (1,N) register algorithm for multi-processors has been provided by Lamport in~\cite{Lam77}. This solution enables
wait-free writes, but only guarantees lock-free read operations, since the writer
can force slow-running readers to retry their read operations indefinitely. A fully wait-free solution has been presented by
Peterson~\cite{Pet83}, which marked the begin of a long running research path towards the construction of wait-free solutions to the readers/writers problem. Along this path we find proposals dealing with (1,1)~\cite{Lam86,Sim90}, (1,N)~\cite{Hal95,Pet83}, and (M,N) registers~\cite{Pet83,Vit86}. A common aspect that characterizes these proposals is that they build wait-free multi-word registers by relying only on single-word read/write registers, just based on atomic single-word read/write instructions. Thus they do not exploit synchronization facilities offered by conventional multi-processor/multi-core machines, such as RMW instructions like Compare-and-Swap (CAS). The disadvantage lies in that, in order to assess the validity of a multi-word atomic read/write operation, it must be carried out multiple times (e.g., 2 times in~\cite{Pet83}), which may impair performance
%(as well as energy efficiency)
especially when scaling up the size of the register. In our approach we avoid this drawback by avoiding at all multiple copies of the register content upon both read and write operations. In particular, we support write operations with a single copy of the new register content into the target buffer. Also, read operations do not need any intermediate data copy, since the reading process can directly read data from the buffer originally targeted by the write operation that is serialized before the read itself. Hence, in ARC, accessing the register in read mode only entails retrieving the correct buffer address.

Several proposals~\cite{And95,Fat11,Agh14,Zhu15} allow to realize a wait-free register by relying on a wait-free universal construct~\cite{Her88}.
This is a design choice that we have explicitly avoided, making our proposal mostly orthogonal. In fact, the employment of a universal construct does not allow capturing the intrinsic properties of the different register operations (read vs write). In turn, this might reduce performance since the number of synchronization steps  might be much larger than what strictly required (just depending on the different nature of the operations). As an example, the work in~\cite{Zhu15} realizes a read operation as a generic one, making it at least as heavyweight as a write operation, while in ARC we have explicitly differentiated the implementations of read and write operations, so as to jointly optimize their execution path.
Moreover, a number of synchronization steps not adhering to the minimum required might have a negative impact on
%both
scalability
%and energy efficiency
of the overall algorithm also because of the effects on the underlying memory hierarchy, which might be further amplified by cache-unaligned data structures.

An additional difference between ARC and the work in~\cite{Zhu15} lies in the fact that the space requirement of the latter for the wait-free implementation is $O(N^2)$ buffers, while we stick to the traditional lower bound of $N+2$ buffers. A quadratic memory cost is also paid by the proposal in~\cite{Agh14}, which additionally has an $O(N)$
%---$n$ being the number of concurrent processors---
time due to the reliance on hazard pointers. We pay such a linear cost only in some corner cases of a write operation, since ARC provides constant-time for reads and amortize constant-time for writes.

Among the aforementioned works exploiting the concept of universal constructor,~\cite{Fat11} is the only one using RMW instructions. Nevertheless, wait-freedom is guaranteed by having all threads register the operation that they want to do---either a read or a write operation---in a shared buffer. Then, all the threads attempt at the same time to complete all the registered operations, ensuring that only one of them actually succeeds. This implies a total of $O(N^2)$ attempts to carry out $N$ operations concurrently executed by the threads. On the other hand, we keep the wait-free nature of the algorithm, while avoiding to have multiple threads carry out same operations.

To the best of our knowledge, the only other proposal based on RMW instructions offered by the underlying computing platform to support an atomic wait-free (1,N) register is the one in~\cite{Lar09}. Here the authors use 64-bit atomic memory operations to update/retrieve a bit-mask indicating what is the buffer instance containing the updated version of the register content and which threads are reading this content version.
%So the 64-bit mask is used to keep track of both the buffer instance and (possible) standing reads on that buffer (and by which threads).
The overall number of slots to be managed is $N+2$, as a classical minimum requirement for a wait-free (1,N) register. Hence, by partitioning the 64-bit mask into the two aforementioned portions (one for the buffer instance and the other for standing reads identification), the maximum number of admitted concurrent readers is 58. Compared to this approach, we use RMW instructions on 64-bit words in a completely different manner, since we do not associate individual bits with threads (to indicate whether a given thread has a standing read on a given buffer instance). Rather, we adopt an anonymous scheme where registering a thread as a reader of a given buffer instance only entails incrementing a per-instance counter of standing reads (hence the name ARC for our proposal).
As a consequence, we can host up to $2^{32}-2$ concurrent readers, which is done by still relying on $N+2$ buffers to keep the register content. Overall, compared to the work in~\cite{Lar09}, our proposal handles scenarios with a large/huge increase of the amount of threads allowed to concurrently perform read operations. Hence, we enable scaled-up wait-free concurrency on the atomic (1,N) register up to a level  fitting the requirements of massively parallel applications hosted by huge (virtualized) hardware parallel platforms. Also, as hinted in the introduction section, the actual number of RMW instructions executed in our register algorithm under diverse workloads is typically lower than the one of~\cite{Lar09}. As we will show via experimental data, this leads to a reduced impact of platform-level synchronization on performance by our proposal.

\section{Anonymous Readers Counting}
\label{sec:technical}

\subsection{Basics}
A \textit{multi-word shared register} is an abstract data structure that is shared by a number of concurrent processes\footnote{From now on we use the term `process' and `thread' interchangeably since the classical literature on register algorithms uses the term `process' to indicate the active entity that can operate on the register.}~\cite{Lam77,Pet83}.
%\footnote{From now on we use the term `process' and `thread' interchangeably since the classical literature on register algorithms uses the term `process' to indicate the active entity that can operate on the register.}~
Each process is allowed to perform two operations on the register: a \textit{read}, which retrieves the most up-to-date value kept by the register, and a \textit{write}, which stores a new register's value. We consider \textit{asynchronous} processes, meaning that no assumption is made on their relative speed or on the interleave of their operations. The operations by a same process are assumed to execute sequentially.
The weakest class to which a register can belong is that of \textit{safe} registers~\cite{Lam86}. A register is safe if its correct value can be always retrieved only in case no concurrency is allowed among reads and writes. Considering that we target concurrent objects, we consider a stronger class, namely \textit{regular} registers.

Regular registers are defined in terms of possible \textit{execution histories} of concurrent read/write operations. In particular, each operation $O$ on the register has a wall-clock time duration, which can be denoted as $[O_s, O_e]$ where $O_s$ and $O_e$ are the starting and ending instants, respectively. A regular register is one that is safe, and in which a read operation that overlaps (in time) a series of write operations obtains either  the register value before the first of these writes or one of the values being written~\cite{Lam86}. Introducing a reading function $\pi$ to assign a write $w$ on the register to each read $r$ such that the value returned by $r$ is the value written by $w$,
and defining a precedence relation on the operations leading to a strict partial order `$\to$'~\cite{Lam86},
a regular register always respects the following property:
\begin{itemize}[noitemsep]
\item \textbf{No-past}. There exists no read r and write w such that $\pi$(r) $\to$ w $\to$ r.
\end{itemize}

Anyhow, a regular register does not ensure that multiple reads executed concurrently to a write must ``agree'' to the same value. By the \textit{linearizability} property~\cite{Her87,Her90}, we can always find a \textit{linearization point} which provides the illusion that each operation $O$ takes effect instantaneously at some point between $O_s$ and $O_e$. Consequently, a stronger class of registers is that of \textit{atomic} registers, defined according to the following criterion~\cite{Lar09}:
\begin{criterion}
A shared register is atomic iff it is regular and the following condition holds for all possible executions:

\begin{itemize}[noitemsep]
%\item \textbf{No-irrelevant}. There exists no read r such that r $\to$ $\pi$(r).
%\item \textbf{No-past}. There exists no read r and write w such that $\pi$(r) $\to$ w $\to$ r.
\item \textbf{No New-Old inversion}. There exist no reads r1 and r2 such that r1 $\to$ r2 and $\pi$(r2) $\to$ $\pi$(r1).
\end{itemize}
\end{criterion}

With an atomic register, reads can be separated among those ``happening'' before and after the linearization point of some write. This categorization marks the difference among the concurrent reads that can return the old value and those which need to return the new value. In particular, if two reads from the same process overlap a write then the later read cannot return the old value if the earlier read returns the new one. Atomic registers have been shown to be linearizable~\cite{Her08}.

\subsection{Memory consistency model}

Multi-processor/multi-core shared memory systems, which we target in this article, offer \textit{memory consistency models}~\cite{Sor11} as kind of ``contracts'' among software developers and hardware manufacturers. They discriminate what software can expect to be guaranteed by the underlying hardware.
A variety of consistency models exist, which are often presented as a set of rules. The simplest memory consistency model is \textit{sequential consistency}. In this model \enquote{the results of any execution is the same as if the operations of all the processors were executed in some sequential order, and the operations of each individual processor appear in this sequence in the order specified by its program}~\cite{Lam79}. This model ensures that all read and write instructions executed by any processor are observed in the same order by all the processors in the system. Peterson's algorithm~\cite{Pet83} and several lock-based algorithms~\cite{Bar93,Her90} require sequential consistency to prove correct.

We assume a weaker consistency model, namely \textit{Total Store Order} (TSO)~\cite{Sor11}, which is used by most off-the-shelf platforms, such as SPARC and x86, thus making our solution of general practical applicability. With TSO,  CPU-cores usually use \textit{store buffers} to hold the stores committed by the overlying pipeline until the underlying memory hierarchy is able to process them. In particular, a store leaves the buffer whenever the cache line to be written is in a coherence state such that the update can be safely performed. TSO allows what is called a \textit{store bypass}: even if a CPU-core outputs a write before a read, their order on memory (as seen by other CPU-cores) can be reversed.

While TSO produces no damage in many applications (rather, it can provide a significant speedup due to a reduced latency on the memory hierarchy), synchronization based on shared memory data must explicitly cope with this scenario.
%In particular, there are explicit synchronization actions which should be treated differently from ordinary memory accesses.
In fact, store bypasses can affect the correctness of synchronization algorithms (e.g. register algorithms) for concurrent processes only relying on individual read/write operations (just like \cite{Pet83}). On the other hand, TSO-based architectures offer particular instructions, referred to as \textit{memory fences}, which enable recovering sequential consistency by explicitly flushing store buffers before executing any other memory operation, thus allowing to preserve the ordering across subsequent read\slash write operations.

Still, for scenarios where synchronization among processes requires to atomically perform pairs (or more) operations, memory fences do not suffice.
To cope with this issue, TSO-based conventional architectures offer \textit{Read-Modify-Write} (RMW) instructions,
whose execution directly interacts with cache controllers so as to ensure that cache lines keeping \textit{synchronization variables} are held in an exclusive state until a couple of read\slash write operations are executed atomically~\cite{Sor11}. This means that no other cache can keep the same line in read mode until the couple of operations completes.
Classical RMW instructions supported by off-the-shelf processors are: \textit{load-link\slash store-conditional} or \textit{compare and swap}, which (although in a different manner) update a memory location only if its content was not changed since the last read access;
\textit{atomic exchange}, which atomically reads the content of a memory location and updates its value;
\textit{add and fetch}, which increments a memory location and reads the updated value.
%\textit{fetch and add}, which reads the value of a memory location and increments its content;
%\textit{fetch and or}, which reads a memory location and computes the bitwise OR with the its content and a specified value;

\subsection{The Register Algorithm}
Similarly to most wait-free (1,N) registers, in our algorithm we use $N+2$ buffers to keep different snapshots of the register value, as produced along time by write operations. This allows each reader to keep a buffer for reading (possibly different across the $N$ readers), while at least 2 buffers are still available to keep some up-to-date register value (the one written while  the readers were concurrently reading the register) and the work-in-progress copy being produced by the writer, if any.

The core data structure we exploit is a single-word shared synchronization variable called {\tt current}. It is a 64-bit shared variable divided into two fields: {\sf index}, keeping the index of the slot containing the most up-to-date register value, and {\sf counter}, namely the readers' presence counter (the number of standing concurrent reads on the slot targeted by {\tt index}). {\sf index} is 32 bits wide, so up to $2^{32}-2$ concurrent readers are allowed%
\footnote{We have selected 32 as a meaningful value for common off-the-shelf architectures which use 64-bit words and RMW instructions targeting 64-bit memory locations. In different architectures, this could be set to an even larger value, by simply having the {\tt current} variable enlarged in size, depending on the actual size of memory locations targeted by RMW instructions.}.

Additionally, our register data structure is made up by $N+2$  slots forming an array which we refer to as {\tt register[]}.
Each slot of this array is an instance of a data structure containing the following fields:

\begin{labeling}{contenttt}
	\item[{\tt r\_start}] The number of read operations started on the slot since its last update.	
	\item[{\tt r\_end}] The number of read operations completed on the slot since its last update.
	\item[{\tt size}] The size of the register value stored in the slot.
	\item[{\tt content}] A pointer to the memory location (the buffer) where the register content is stored.
\end{labeling}

The {\tt size} field is introduced since we support writes (hence reads) of different sizes, meaning that each register value can have a different size (clearly up to some admissible maximum). Also, with no loss of generality, while presenting the register pseudo-code we assume that the buffer pointed by the {\tt content} field of the register slot is already allocated, and that it can host the maximum sized register content (depending on the usage scenario). In any real implementation of our register algorithm, dynamic buffer allocation/release, with each buffer made up by the amount of bytes fitting the size of the register value to be stored upon write operations could be employed.

\begin{algorithm}[!t]
\caption{Register initialization.}
\label{alg:init}
\begin{algorithmic}[1]
{\small
\Procedure{INIT}{$content$, $size$}
\ForAll {$slot \in [0,N+1]$}
	\State $register[slot].size \gets 0$
	\State $register[slot].r\_start \gets 0$
	\State $register[slot].r\_end \gets 0$
\EndFor
\State \Call{MemCopy}{$register[0].content,content,size}$
\State $register[0].size \gets size$
\State $current \gets N$\label{alg:init:virtual} \Comment{\textbf{I1}}
\EndProcedure
}
\end{algorithmic}
\end{algorithm}

The initial setup of the register data structure is shown in Algorithm~\ref{alg:init}.
With no loss of generality, we assume that the register is initialized to keep its initial value into {\tt register[0]}, and that all the other $N+1$ entries are all available for posting some new register value.

Algorithm~\ref{alg:read} shows the pseudo-code for the read operation. By exploiting the {\tt AtomicAddAndFetch}  instruction targeting {\tt current}, a reader process is able to atomically retrieve the index of the slot containing the most up-to-date register value and increment the corresponding presence counter (\textbf{R4}).
This allows us to enforce \textit{visible reads}~\cite{Bur80}, although we do this in an \textit{anonymous way}. In fact, the presence counter is not used to indicate who has started reading the up-to-date register value, rather how many processes did it. The index of the slot from which the up-to-date value is to be found is extracted by executing bitwise instructions on the value returned by {\tt AtomicAddAndFetch}.

\begin{algorithm}[!t]
\caption{The atomic register read operation.}
\label{alg:read}
\begin{algorithmic}[1]
{\small
\Procedure{Read}{\,}
	\State $index \gets current \gg$ 32 \Comment{\textbf{R1}}
	\If{$last\_index = index$}
		\State $entry \gets register[last\_index]$
		 \State \Return{$\langle entry.content, entry.size \rangle$} \Comment{\textbf{R2}}
	\EndIf
	\State \Call{AtomicInc}{$register[last\_index].r\_end$} \Comment{\textbf{R3}}
	\State $tmp\_curr \gets $ \Call{AtomicAddAndFetch} {$current, 1$} \Comment{\textbf{R4}}
	\State $last\_index \gets $ $tmp\_curr \gg 32$ \Comment{\textbf{R5}}
	\State $entry \gets register[last\_index]$
	\State \Return{$\langle entry.content, entry.size \rangle$}
\EndProcedure
}
\end{algorithmic}
\end{algorithm}

We consider a read operation from a slot as concluded as soon as the reader tries to read again from the register. When, this happens, the {\tt r\_end} counter of the slot from which the reader took the register value upon its last read is incremented atomically. A special case occurs when the already-read slot still keeps the most up-to-date register value (\textbf{R2}). In this case, {\tt r\_end} is not incremented to indicate that the reader did not conclude its operations on the slot yet---a new read is just starting, bound to that same slot. Incrementing {\tt r\_end} only when moving to another slot (upon a subsequent read that finds a newer register value) allows us to avoid overflows of counter variables (\textbf{R3}). Thus we enable an infinite number of reads (by any reader) to occur on a slot that still keeps the up-to-date register value. In order to remember from which slot the reader took the register value upon its last read we use the {\tt last\_index} variable (which is local to a reader), where we load the index of the target slot for the read operation each time the reader accesses a newer register value (\textbf{R5}). The check on whether the last accessed register value is still the most up-to-date is executed by loading the index kept by {\tt current} (\textbf{R1}) as soon as the read operation starts, and then comparing it with {\tt last\_index}. Given that the value of {\tt current} is manipulated by any process (including the writer, as we will show) via RMW instructions only, then the index value returned by reading {\tt current} (\textbf{R1}) is guaranteed to represent a correct snapshot of the shared synchronization variable we use in our register algorithm under the assumed TSO memory consistency model.

%Upon initialization, the
As startup {\tt current} is initialized to $N$ (\textbf{I1}). This sets its most-significant 32 bits (the {\sf index} sub-field) to zero and initializes the {\sf counter} sub-field as if all the readers had already started reading from the 0-th (initially-valid) slot. Therefore, if no update is ever made on the register's content, readers will indefinitely read this value (\textbf{R1}).

\begin{algorithm}[t]
\caption{The atomic register write operation.}
\label{alg:write}
\begin{algorithmic}[1]
{\small
\Procedure{Write}{$content$, $size$}
	\State pick $slot$ such that $slot\neq last\_slot \wedge register[slot].r\_start = register[slot].r\_end$ \Comment{\textbf{W1}}
	\State \Call{MemCopy}{$register[slot].content,content,size$}
	\State $register[slot].size \gets size$
	\State $register[slot].r\_start \gets 0$
	\State $register[slot].r\_end \gets 0$
	\State $old\_curr \gets $ \Call{AtomicExchnge} {$current$, $slot \ll$ 32} \Comment{\textbf{W2}}
	\State $old\_slot \gets old\_curr$ $\gg$ 32
	\State $register[old\_slot].r\_start \gets old\_curr$ $\&~(2^{32}-1)$ \Comment{\textbf{W3}}
	\State $last\_slot \gets slot$
\EndProcedure
}
\end{algorithmic}
\end{algorithm}

The pseudo-code for the write operation is shown in Algorithm~\ref{alg:write}.
Upon writing, the writer process selects a free slot, namely a slot which is not currently bound to any not yet finalized read operation by whichever process, and which is different from the slot that was used for the last write operation (say the one kept by {\tt current}). In compliance with the initialization of the register, we assume that the {\tt last\_slot} local variable kept by the writer, indicating the last slot used for a write, is initialized to the value 0. In fact, at init-time the initial register content is posted onto the 0-th slot.  The writer detects if no more processes are currently reading from a slot by checking whether the two counters {\tt r\_start} and {\tt r\_end} associated with the slot keep the same value. The writer then performs a copy operation of the new value to the selected slot, and updates all the fields of the slot entry. In particular, it sets both {\tt r\_start} and {\tt r\_end} to zero, and {\tt size} to the actual size of the new register value that is being stored. Then, by using an {\tt AtomicExchange} instruction (\textbf{W2}), the writer changes the content of the {\tt current} shared synchronization variable so as to ``publish'' the index of the new slot from which readers can start performing read operations. Given that the update of {\tt current} is based on the execution of an RMW instruction, the content of the slot selected for the new write operation is guaranteed to be coherent when the {\tt current} variable is updated under the assumed TSO memory consistency model. In other words, if a reader gets the updated {\tt current} value (\textbf{R4}) and accesses the target slot, the accessed data are guaranteed to be coherent with the corresponding updates performed by the writer.

The new value of {\tt current} which is atomically written by the writer (\textbf{W2}) has a {\sf counter} field set to zero, telling that the new version has not been read by any process yet.
The {\tt AtomicExchange} allows to retrieve as well the old value of {\tt current}, which is loaded into the {\tt old\_current} variable local to the writer. This is used by the writer to extract the old {\sf counter} field, and store its value in the {\tt r\_start} field of the old (the last written) slot (\textbf{W3}). In this way, the number (not the identity) of readers which started an operation on the old slot is ``freezed'' into the slot management meta-data.
We note that, after such freezing takes place for some slot, the corresponding  values {\tt r\_start} and {\tt r\_end} are such that {\tt r\_start} $\geq$ {\tt r\_end}. But eventually these two values will be the same, which is the condition telling the writer that the slot has been released by all the readers since they moved to some fresher slot. In fact, the condition {\tt r\_start} $=$ {\tt r\_end} indicates to the writer that the slot is free again (\textbf{W1}). On the other hand, any written slot that is never accessed by any reader up to the point in time  where some newer register value is atomically published by the writer, will have its  {\tt r\_start} and {\tt r\_end} fields both set to zero, which implies it is a free slot.

\subsection{Speeding up free-slot searches}

By the pseudo-code of ARC read operations can be trivially shown to take constant-time. On the other hand, write operations require searching for a free slot among $N+2$ (\textbf{W1}), which would imply linear time complexity. To provide amortized constant time for write operations (in particular for the slot search operation),
%We also devise an optimization oriented to speeding up the search operation of a free slot by the writer (\textbf{W1}). The search operation is linear in the number of slots of the register.
%However,
readers that complete their read from a slot by incrementing the corresponding {\tt r\_end} counter (i.e. they release the slot), can check whether this counter is equal to the {\tt r\_start} counter associated with the same slot. If this is true, then by the register algorithm structure it means that the slot can be reused for subsequent writes. Hence, a reader detecting such an equality can post into another shared variable the index of the just-released slot. This can be used by the writer as a \textit{proposal} to start searching for a free slot. This proposal will always correspond to an actually free slot (hence enabling constant time retrieval of the free slot upon write operations) except for the corner case where the writer already took the same slot for some already issued write having observed its release before the reader posted its proposal.

\section{Correctness Proof}
\label{sec:proof}

By code construction, all invocations to {\sc Read}() are guaranteed to complete in a finite number of steps. Hence reads are by construction guaranteed to be wait-free. As for the  {\sc Write}() operation, completion within a finite number of steps is guaranteed if the free-slot search operation carried out at the beginning of the write operation completes in a finite number of steps. This is true if it is guaranteed that at least one slot different from the last one used for a register write is in a state such that its {\tt r\_start} and {\tt r\_end} fields are equal. This is proven in the following Lemma:

\begin{lemma}
\label{lemma:lemma0}
Upon starting a write operation at least one of the $N+2$ register slots, which is different from {\tt last\_slot}, is such that {\tt r\_start} and {\tt r\_end} keep the same value.
\end{lemma}

\begin{proof} This proof is based on two disjoint cases analysis:

{\bf Case 1}. The writer performs its first write on the register. In this case, all the {\tt r\_start} and {\tt r\_end} fields are still found to be set to the value 0. This is because no reader could have updated any {\tt r\_end} field in any slot since this can only happen if a newer register value is found upon a read operation, which is not the case since the writer did not yet post any new value, say {\tt current} has never been updated. Also, the writer did not yet update any {\tt r\_start} field in any slot, since this takes place just while performing a write operation, while we are at the beginning of the first write. Given that {\tt last\_slot} is set to 0 upon register initialization, all the $N+1$ slots that are different from the 0-th one are such that that their {\tt r\_start} and {\tt r\_end} fields are both set to zero. Hence at least one slot different from {\tt last\_slot} is such that its {\tt r\_start} and {\tt r\_end} fields keep the same value, and the claim follows.

{\bf Case 2}. The writer performs the $i$-th write on the register. In this case, all writes up to the ($i-1$)-th one have updated {\tt current}, and the readers might have fetched the various values of {\tt current}, also releasing a presence count unit each time this happened. By the {\sc Read}() operation pseudo-code, a reader leaves a presence count unit on some slot (updating the {\tt counter} field of the variable {\tt current}) only after having released a count unit on the {\tt r\_end} field of some other slot. Hence, for all the {\tt r\_start} units freezed by the writer into the slots upon performing writes up to the $(i-1)$-th we have that:
$$\sum_{j=0}^{N+1}{(register[j].r\_start - register[j].r\_end)}\leq N$$

Hence, given that {\tt r\_start} and {\tt r\_end} fields are non-negative values, for at least 2 different slots of the $N+2$ slots of the register, these same fields must have the same value.
Hence at least one slot which is different from {\tt last\_slot} is such that its {\tt r\_start} and {\tt r\_end} fields keep the same value, and the claim follows.
%\qed
\end{proof}

%On the other hand,
We now prove consistency of concurrent read/write operations in our register:
%the correctness with respect to atomicity of the register operations. We start by proving the following Lemma:

\begin{lemma}
\label{lemma:lemma1}
While the writer is executing a write operation on a slot, no reader will read the same slot until the write completes.
\end{lemma}

\begin{proof}
%(By contradiction).
%We prove this lemma by contradiction.
%Let us assume that a read operation and a write operation overlap on the same slot.
Read operations bound to the initial snapshot of the register trivially satisfy the claim, since that snapshot is not written by the writer. Let us therefore focus on reads of the register snapshots that are different from the initialization one.
By the {\sc Read}() operation pseudo-code,
%and by having the local variable {\tt last\_index} of any reader initialized to the {\tt NO\_SLOT} value,
a read operation is always bound to the slot index that is returned at some point in time by atomically executing {\tt AtomicAddAndFetch} on the {\tt current} variable. This is true also when subsequent reads by a reader process take an unchanged register content from a same slot, since the first of these reads must have necessarily executed the {\tt AtomicAddAndFetch} instruction on {\tt current} to retrieve the index of that slot. On the other hand, the {\tt r\_end} field of some slot is incremented by the reader only after moving to some new slot upon the read operations.

Given that (i) the writer selects a slot $x$ for writing only when it finds its {\tt r\_start} and {\tt r\_end} fields set to the same value, (ii) {\tt r\_start} is freezed into the slot $x$  only after it is no longer the current one, (iii) whichever slot $x$ becomes again readable after its index is published into the {\tt current} shared variable, (iv) TSO memory consistency guarantees that when the update of {\tt current} is performed by the writer, so as to point to the $x$-th slot, all the data associated with the slot have already been flushed to memory, we have that any read will always observe a stable snapshot of the register when reading from the generic $x$-th slot. Hence the claim follows.
%\qed
\end{proof}

We now prove  regularity and atomicity of our register:

\begin{theorem}[Regular Register]
\label{thm:theo1}
Any read operation returns either the last written value, or one being concurrently written.
\end{theorem}

\begin{proof}
By the structure of Algorithm~\ref{alg:write}, the update of {\tt current} (\textbf{W2}) represents the atomic memory operation that defines the linearization point for any write.
If the write is linearized before the execution of statement \textbf{R1} by some reader, a read always returns the last written value, say the one posted by the last write serialized before the read, since the serialization point of the read is determined by \textbf{R4}, which targets the same shared synchronization variable {\tt current} whose atomic updates represent the serialization points of writes. Otherwise, if \textbf{R1} is executed before the update of the {\tt current} shared synchronization variable by the write operation, the read is correctly allowed to return the register value that was already stored before any concurrent write. Hence the claim follows.
%\qed
\end{proof}

\begin{theorem}
\label{thm:theo2}
Given two read operations $r_1$ and $r_2$ such that $r_1 \rightarrow r_2$, $r_2$ never returns a value older than the one returned by $r_1$.
\end{theorem}

\begin{proof} (By contradiction)
By the proof of Theorem~\ref{thm:theo1}, a read executed before the linearization point of a write returns the old value (with respect to the execution of the write).
Let us assume by contradiction that, given two reads  $r_1$ and $r_2$ such that $r_1 \rightarrow r_2$, $r_2$ returns a value older than the one returned by $r_1$. Yet, the {\tt current} synchronization variable is updated whenever the index of the most up-to-date slot changes. Therefore, for $r_2$ to read a value older than $r_1$, it has to read {\tt current} before $r_1$. But this violates the precedence $r_1 \rightarrow r_2$. Hence the assumption is contradicted and the claim follows.
%\qed
\end{proof}

Atomicity of our register algorithm trivially follows from Theorem~\ref{thm:theo1} and Theorem~\ref{thm:theo2} in combination.

\section{Experimental Results}
\label{sec:experimental}

\begin{figure*}[t]
\centering
\subfigure[4KB register size]{\includegraphics[width=0.6\linewidth]{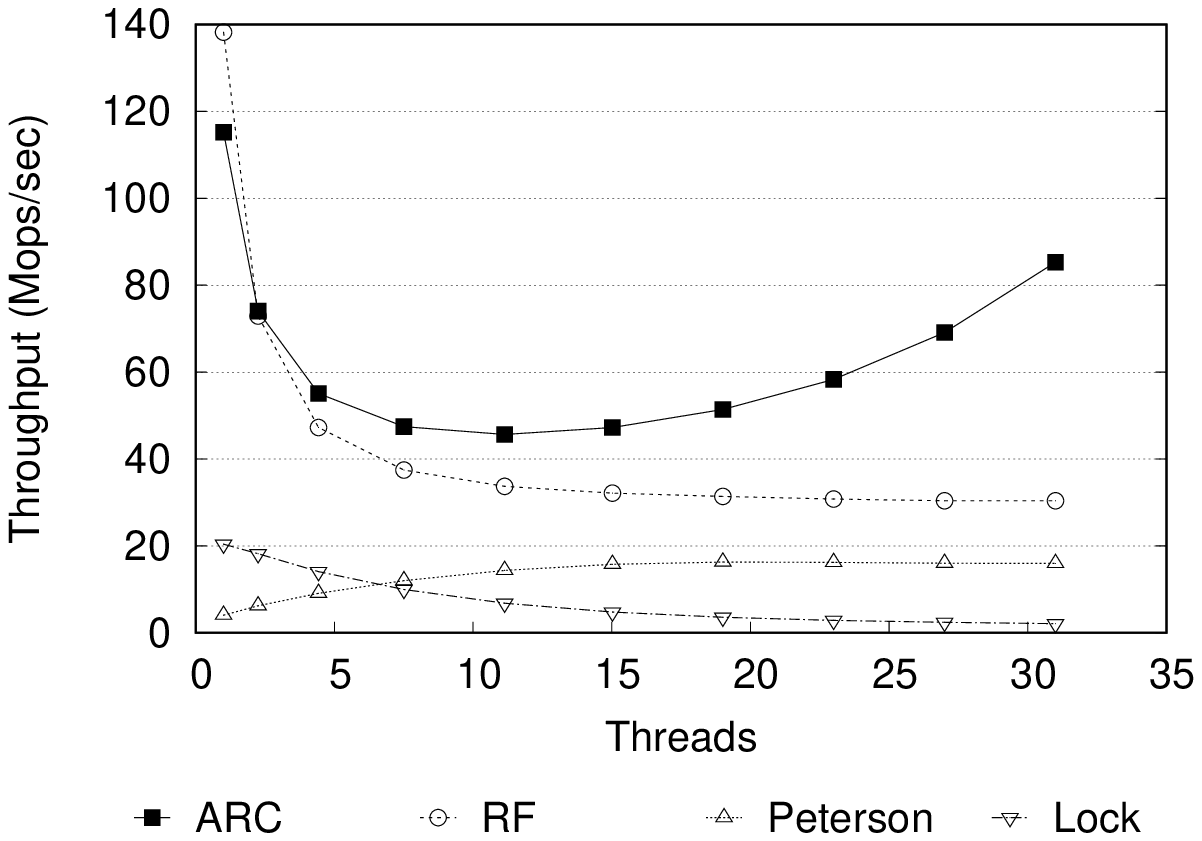}}

\subfigure[32KB register size]{\includegraphics[width=0.6\linewidth]{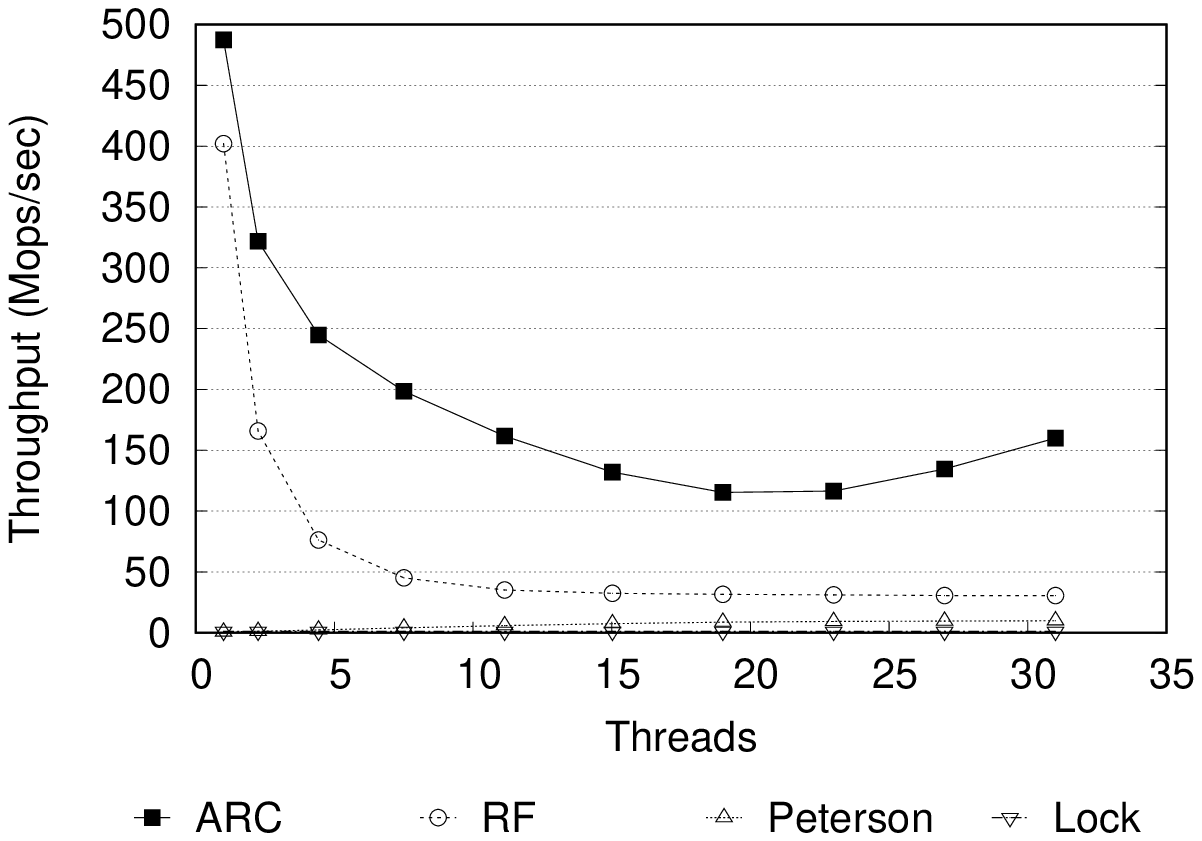}}

\subfigure[128KB register size]{\includegraphics[width=0.6\linewidth]{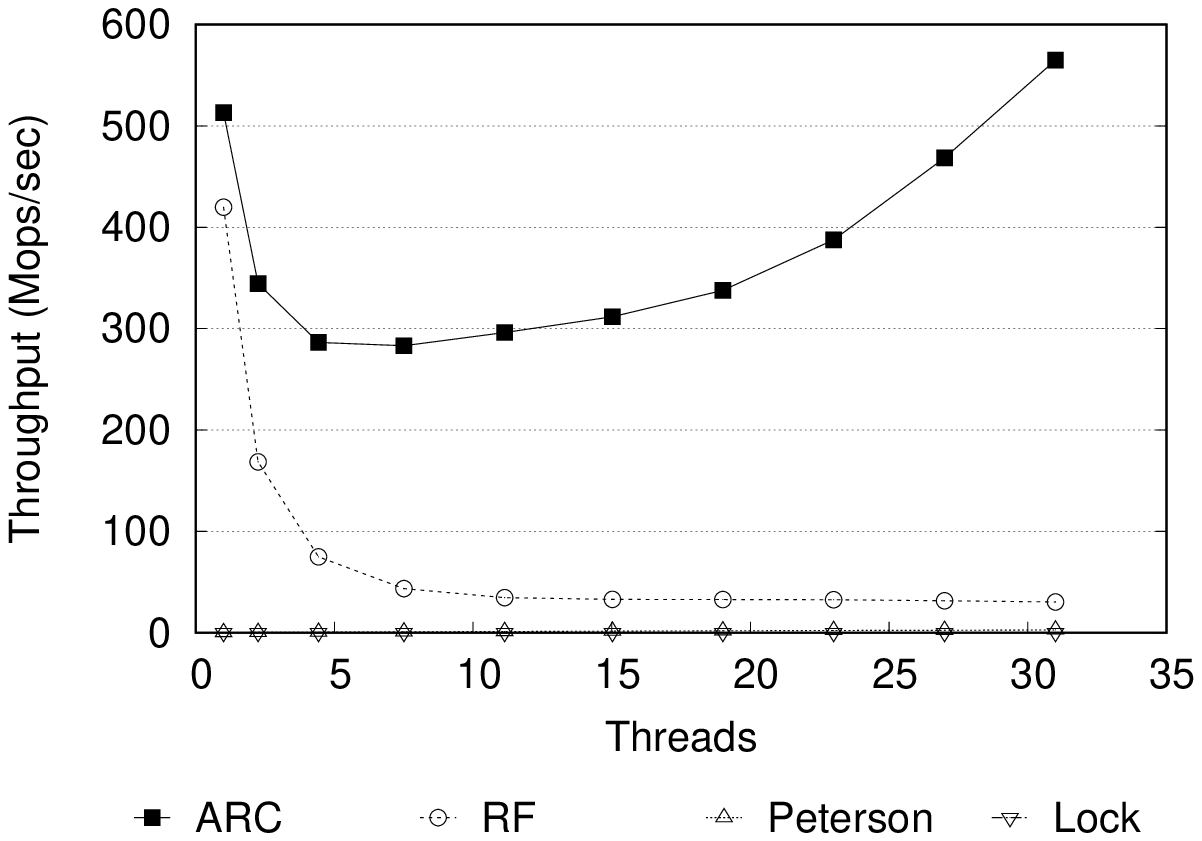}}
\caption{Throughput with different register size values (32 CPU-core physical machine).}
\label{bare-metal}
\end{figure*}

In this section we present a comparative performance study of our ARC algorithm with some literature proposal. In particular, we selected the Readers-Field (RF) wait-free algorithm presented in \cite{Lar09}, still based on RMW instructions offered by the underlying architecture, and Peterson's wait-free algorithm~\cite{Pet83}. For completeness of the analysis we also include a classical lock-based approach (using read/write spin-locks still implemented using RMW instructions) not ensuring wait-freedom. All these algorithms have been implemented according to their original specification by relying on the {\sf C} programming language and UNIX API, plus the nesting of either RMW machine instructions used to manipulate synchronization variables (like in ARC and RF) or memory-fence instructions to guarantee correctness under TSO (like for Peterson's algorithm). Also, in all the implementations we relied on {\tt mmap()} pre-allocation of all the buffers requested by each algorithm. The source code for all the tested implementations is available for free download%
\footnote{Code available at \url{https://github.com/HPDCS/ARC}.}.
In all the implementations, the ``process entity'', encapsulating the sequence of read or write operations accessing the register, is instantiated via an individual thread, scheduled for execution under the control of the operating system.

\begin{figure*}[t]

\centering
\subfigure[4KB register size]{\includegraphics[width=0.6\linewidth]{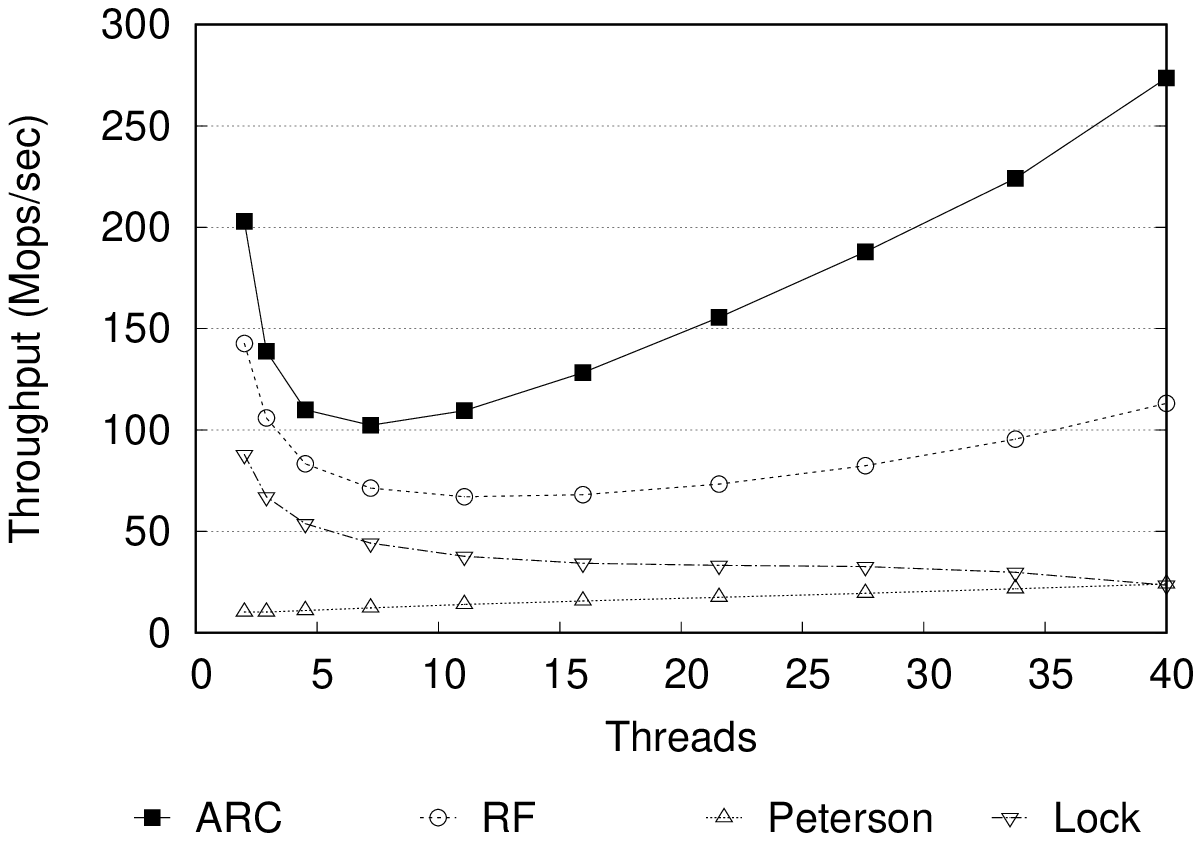}}

\subfigure[32KB register size]{\includegraphics[width=0.6\linewidth]{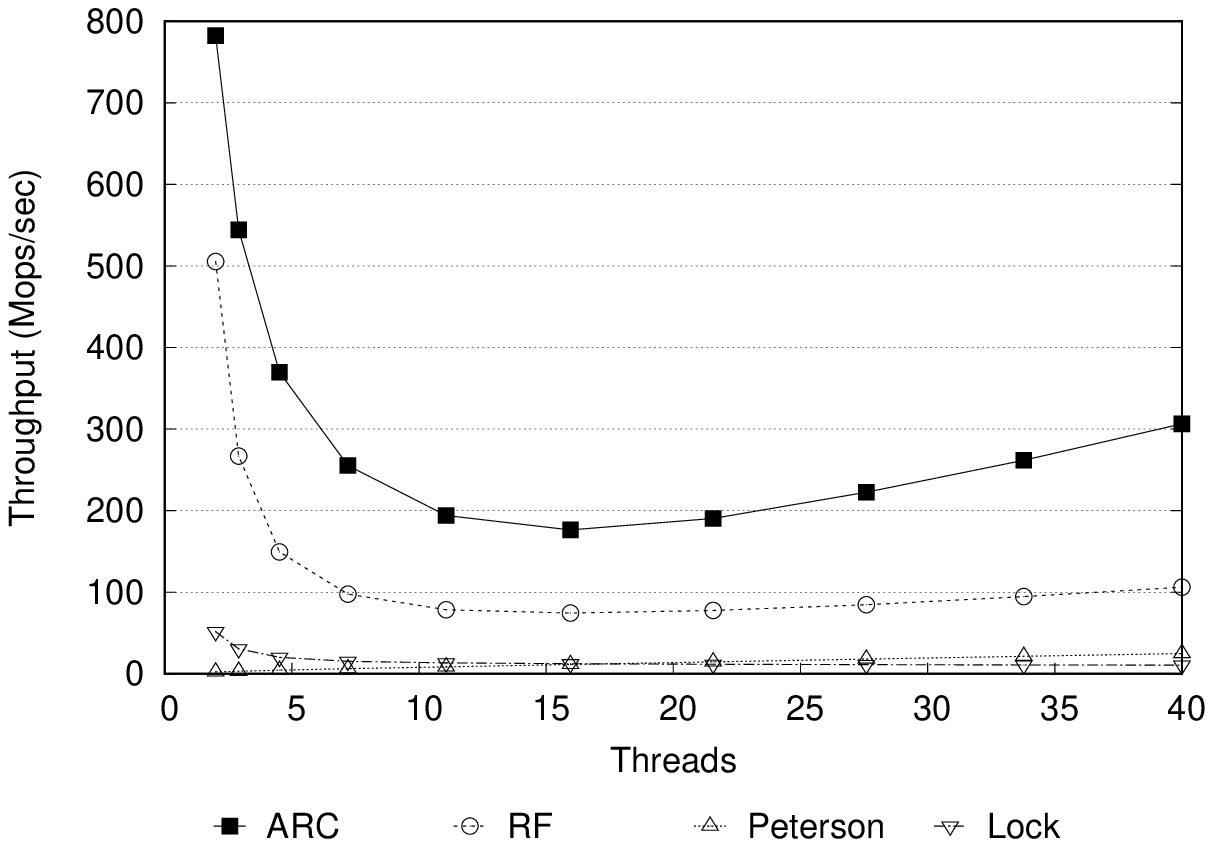}}

\subfigure[128KB register size]{\includegraphics[width=0.6\linewidth]{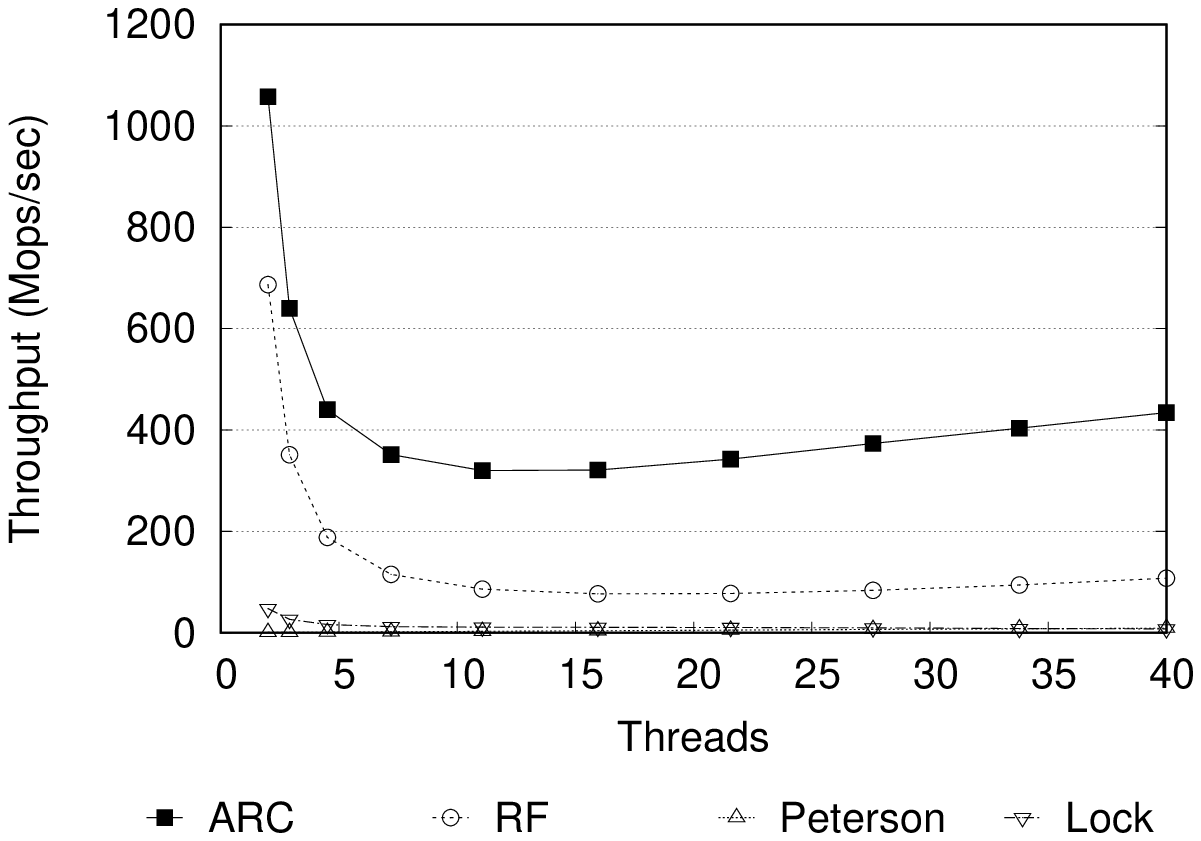}}
\caption{Throughput with different register size values (40 vCPUs machine).}
\label{virtual}
\end{figure*}

We tested the different algorithms deploying their implementations on two different computing platforms, a physical one and a virtualized one. The former is an HP ProLiant sever  equipped with four 2GHz AMD Opteron 6128 processors and 64 GB of RAM. Each processor has 8 cores, for a total of 32 CPU-cores, which share a 12MB L3 cache (6 MB per each 4-cores set), and each core has a 512KB private L2 cache. The operating system is 64-bit Debian 6, with Linux Kernel 2.6.32.5. The virtualized platform is an Amazon m4.10xlarge instance equipped with 2.4 GHz Intel Xeon E5-2676 v3 (Haswell) processors offering a total capacity of 40 vCPUs, equipped with 160 GB of RAM. This virtual machine runs Ubuntu Server 14.04 LTS as the operating system, with Linux Kernel 4.2.

We have conducted two different sets of experimental tests. In the first set, we have generated a workload on the register which is similar in spirit to the well-known {\em Hold Model}~\cite{Vau75}. In particular, all concurrent threads execute operations repeatedly on the register data structure. This means that read and write operations are actually ``dummy'' operations which only execute the ARC algorithms discussed in the previous section (or those of the competitors of ARC)---each write operation simply copies a same content to the register, and a read operation only retrieves the pointer to the valid register buffer. This is an extreme scenario in which data processing has zero latency, and threads make no other work than accessing the register data structure. The effect of this behaviour is that the logical contention on the register data structure and the physical one on the underlying hardware architecture is maximal. This part of the experimentation allows us to assess the benefits of ARC's optimized synchronization strategy. In a second part,
 %of the experimentation,
 we have associated read and write operations with actual processing---a write actually generates some data, and a read scans the whole content of the retrieved buffer. In this second scenario we can study the effect of different operations' latencies.

Before discussing performance results, we recall again that
%it is important to note that our ARC algorithm
ARC and RF not only differ by the different amounts of readers they can handle---58 in RF vs $2^{32}-2$ in ARC. Rather, they also differ by the way RMW instructions are exploited along the execution path of read/write operations accessing the register. This aspect makes a comparative analysis of these two algorithms
%when deployed on the 32 CPU-cores real machine or the 40 vCPUs machine
interesting independently of the huge scale up of the readers count admitted by our ARC proposal.

\begin{figure*}[t]

\centering
\subfigure[4KB register size]{\includegraphics[width=0.6\linewidth]{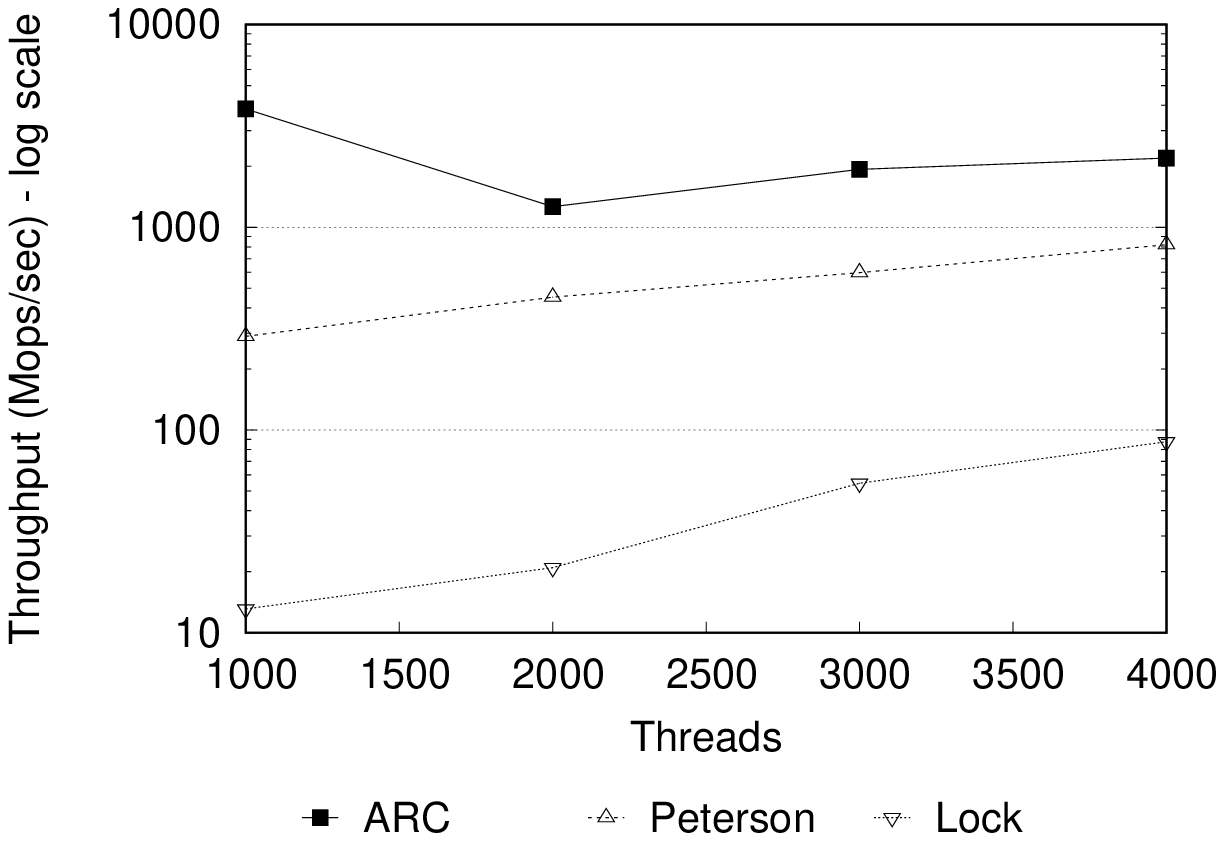}}

\subfigure[32KB register size]{\includegraphics[width=0.6\linewidth]{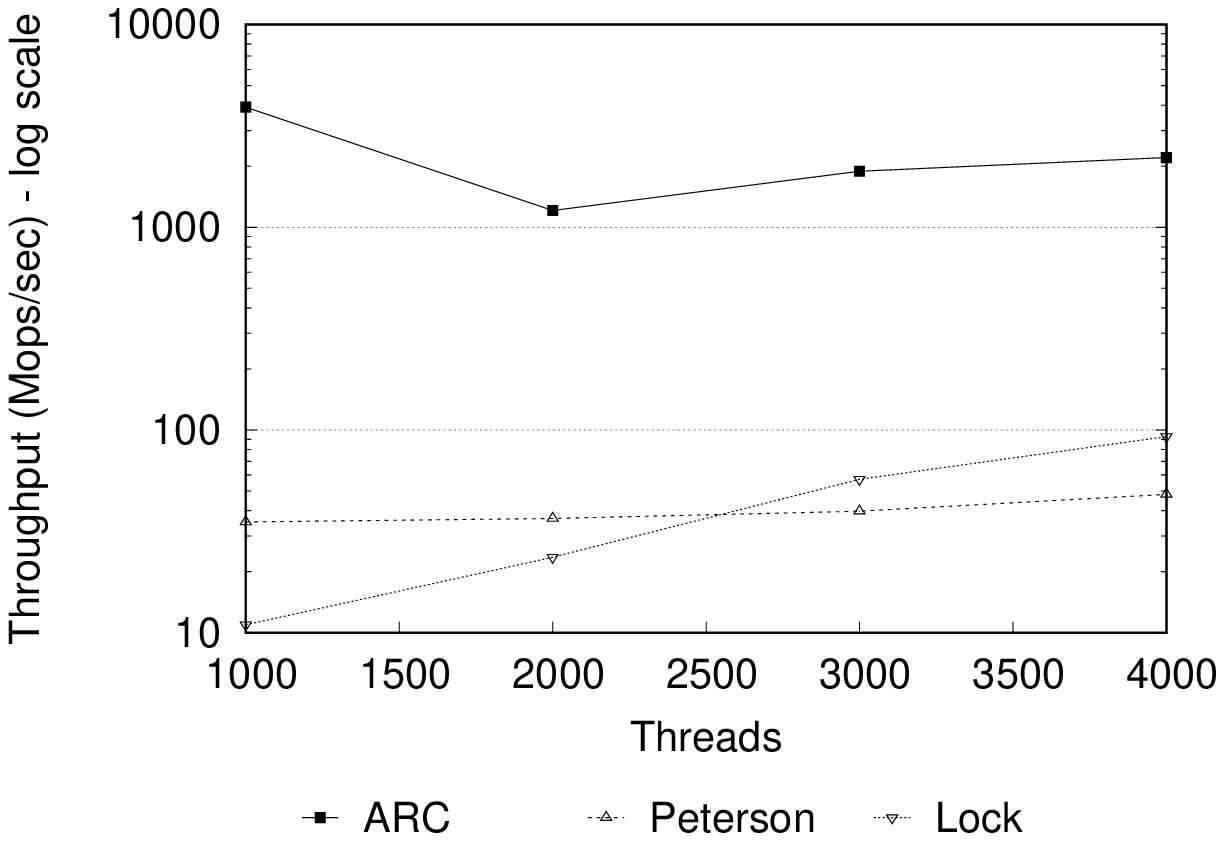}}

\subfigure[128KB register size]{\includegraphics[width=0.6\linewidth]{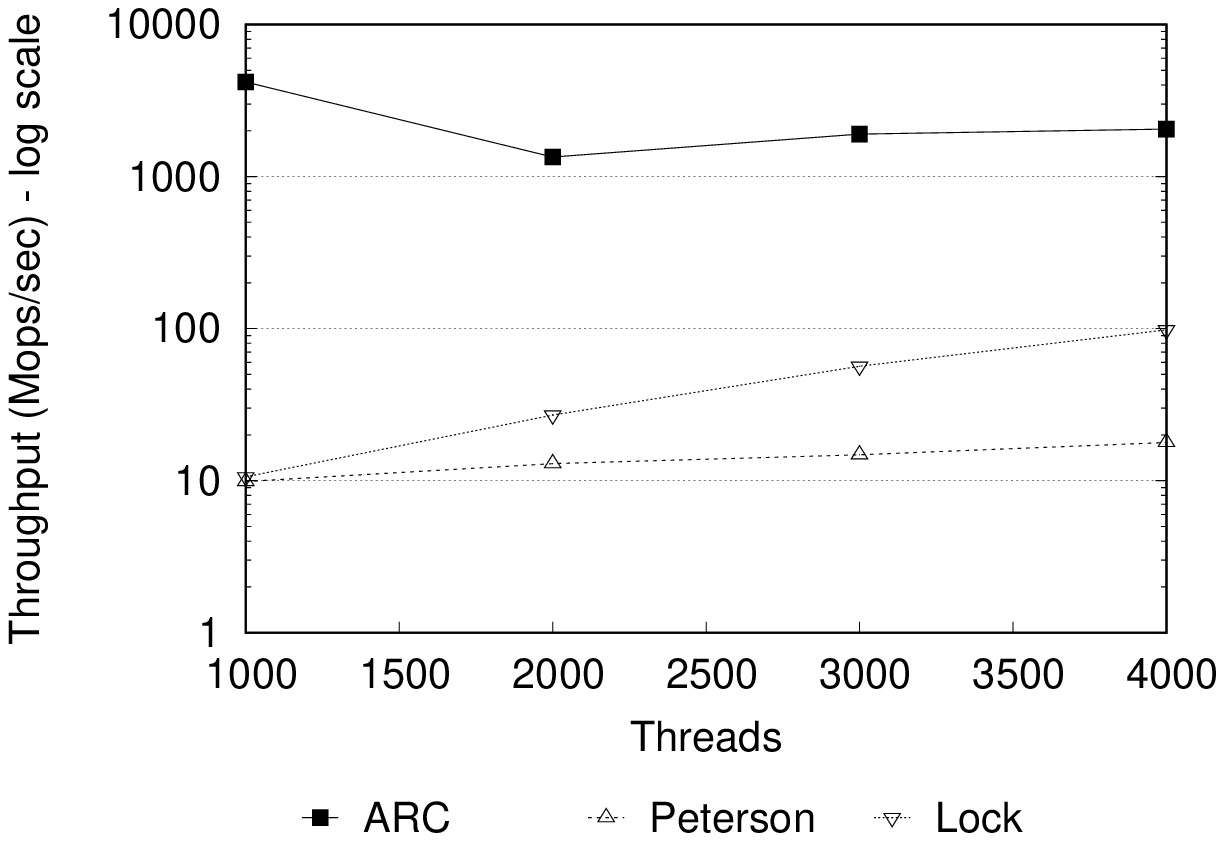}}
\caption{Throughput with largely-increased thread counts (32 CPU-core physical machine).}
\label{time-sharing}
\end{figure*}

In Figure~\ref{bare-metal} we report throughput values (read/write operations per time unit) while varying the number of threads for deploys of the different register implementations on the 32 CPU-core physical machine. In these tests, one thread continuously executes write operations on the register, while all the others continuously execute read operations on the register. Each reported sample is the average over 10 runs, with each run made up by at least $2 \times 10^6$ read/write operations.  No other workload has been activated on the machine, so we are in the scenario of maximal concurrency in the access to the register by the active threads, up to the maximum count of 32 threads that can be hosted by the machine while still avoiding time-sharing concurrency (hence interference) among them on a same CPU-core. Also, the different plots refer to 3 different sizes of the register, a minimal size of 4KB (a single operating system page), an intermediate size of 32KB, and a large size of 128KB. By the plots we see how both ARC and RF outperform the other solutions at any thread count. Also, ARC outperforms RF as soon as the thread count is increased beyond the value 8 or the register size is non-minimal, providing up to an order of magnitude better throughput. The reason for this behavior is that RF executes an RMW instruction (i.e. a {\tt FetchAndOr}) upon any read, while our proposal ARC executes a RMW instruction only if the write operation of a newer register value is serialized before the execution of the statement \textbf{R1} of the read operation in Algorithm \ref{alg:read}. Hence, ARC is more efficient (since it avoids the execution of RMW instructions) upon reading a register content that is still valid (i.e. it did not change since the last read operation executed by the same thread). This scenario shows up when increasing the level of concurrency of read operations or when the write operation takes longer time due to the larger size of the register content to be posted by the writer---we recall that a memory copy is executed upon a write. In both cases more threads will likely find a not-yet-updated register value upon subsequent reads, a scenario which is captured more efficiently by ARC, compared to RF, just avoiding the execution of RMW instructions.

In Figure~\ref{virtual} we report the throughput values that have been observed when running on top of the virtualized platform with 40 vCPUs. These data confirm what we already saw for executions on the physical machine, with the additional indication that ARC performs better than RF even with minimal thread counts and minimal register size, which indicates how the avoidance of the execution of RMW instructions upon read operations in scenarios where newer writes were not serialized before the reads allows to favor performance even more than what happens with deploys on the 32 CPU-core physical machine. Moreover, with respect to the execution on a physical machine, all wait-free algorithms provide a non-negligible performance speedup over the lock-based implementation. In particular, in the best case we have a performance gain which is 2x the one observed on the physical architecture. This is an indication of the benefits which can be obtained when using wait-free synchronization on virtualized architectures. Indeed, lock-based implementations can introduce an additional slow down whenever the virtualized architecture reduces the computing power allocated to the core holding the lock, due to CPU stealing by the underlying hypervisor's hardware architecture.

In Figure~\ref{time-sharing} we report throughput data when running on the 32 CPU-core physical machine with a definitely scaled up thread count (up to 4000). In this scenario, RF could not be tested (since, as said, it supports 58 reader threads only). However, this test settings helped us to assess the performance by ARC compared to Peterson's algorithm and to the lock-based one when considering time-sharing concurrency among the threads, hence interference among them because of competition on CPU usage. By that data we see that both ARC and the lock-based algorithm are not sensible to the increase of the threads count and to the increase of the register size, even though ARC provides orders of magnitude better throughput. This is not true for Peterson's algorithm since it is based on multiple copies when performing access operations. Such multiple copies are clearly adverse to performance in time-sharing concurrency deploys due to highly negative effects on locality and caching efficiency, especially for larger register size.

Overall, ARC delivers better performance than all the tested solutions independently of the type of deploy (physical vs virtual) and of the readers' count or register size. At the same time, it still allows a huge scale-up in the number of readers compared to RF, which appeared to be the best performing literature solution (compared to Peterson's algorithm and the lock-based one) for deploys on the used physical machine.

\section{Conclusions}

In this paper we have presented Anonymous Readers Counting, a multi-word wait-free atomic (1,N) register algorithm targeting shared-memory TSO-consistent parallel architectures. Our register enables up to $2^{32}-2$ readers on 64-bit machines and avoids any intermediate copy of the register content upon any operation, while still using the classical lower bound of $N+2$ buffers for ensuring wait-freedom. It exploits Read-Modify-Write (RMW) instructions commonly supported by off-the-shelf architectures, by also reducing the impact of actually running RMW instructions compared to the reference literature proposal in \cite{Lar09}, which also has the disadvantage of handling up to 58 readers only. The performance benefits from our proposal compared to literature approaches have been shown via a study based on deploys of the compared register implementations on both a parallel physical machine and a virtualized one. We have also provided a proof of correctness of our register algorithm.

%\newpage

%%%%%%%%% -- BIB STYLE AND FILE -- %%%%%%%%
\bibliographystyle{ieeetr}
\bibliography{library}

\
%%%%%%%%%%%%%%%%%%%%%%%%%%%%%%%%%%%%

\end{document}